\documentclass{birkjour}
\title{On Landau-Zener transitions for dephasing Lindbladians}
\author{Martin Fraas}
\address{%
980 Ohlone Ave.\\
Albany\\
CA 94706\\
USA}
\email{martin.fraas@gmail.com}
\author{Lisa H\"{a}nggli}
\address{%
Theoretische Physik\\
ETH Zurich\\
8093 Zurich\\
Switzerland}
\email{haenggli@itp.phys.ethz.ch}
\usepackage{graphicx}
\usepackage{bbold}
\usepackage{amsmath}
\usepackage{mathtools}
\usepackage{color}
\usepackage{amsthm}
\usepackage{enumerate}
\usepackage{todonotes}
\usepackage{cite}
\usepackage{hyperref}
\usepackage{xr}
\newtheorem{theorem}{Theorem}
\newtheorem{lemma}{Lemma}
\newtheorem{proposition}[theorem]{Proposition}

\newcommand{\bra}[1]{\langle #1 \vert}
\newcommand{\ket}[1]{\vert #1 \rangle}

\numberwithin{equation}{section}
\DeclareMathOperator{\tr}{tr}

\DeclareMathOperator{\ran}{ran}
\DeclareMathOperator{\sgn}{sgn}

\begin{document}
\date{\today}

\begin{abstract}
We consider a driven open system whose evolution is described by a Lindbladian. The Lindbladian is assumed to be dephasing and its Hamiltonian part to be given by the Landau-Zener Hamiltonian. We derive a formula for the transition probability which, unlike previous results, extends the Landau-Zener formula to open systems.
\end{abstract}

\maketitle



\section{Introduction}

A dephasing evolution of an open system maps an initial coherent superposition of energy eigenstates to an incoherent mixture of energy states while preserving their populations. In particular, the ground state, or any other energy eigenstate, of a dephasing open system is stationary. An applied driving force will induce transitions in between these states. In this article we discuss the transition probabilities for the case of adiabatic driving.
More precisely, let $p$ be the probability that a system, prepared in the (isolated) ground state at some initial time and evolved to some final time, has left it by then. The evolution considered here is time-dependent, and is of a dephasing type for each fixed time instant. The driving is slow, which is captured by the time parameter $t=\epsilon^{-1} s$, with $\epsilon\rightarrow 0$ in the adiabatic limit. The goal is to derive a formula for $p$ to leading order in $\epsilon$ when the system is driven across an avoided crossing (Landau-Zener transition).
Transition probabilities in open quantum systems have been discussed in other or more general settings, but also with different focus. We mention \cite{FV63,AL87} and \cite{SS93,PS03,SG91}, as general references and references to Landau-Zener transitions respectively. For further references see \cite{AFGG11}; more recent works include \cite{XPV14,AB14}.
In this paper we assume the dynamics of the open system to be Markovian, i.e.\ the environment to be effectively memoryless. The validity of this approximation is for example rigorously proven in the case of weak coupling \cite{EBD74}. In addition, we assume that the Hilbert space of the system is finite dimensional. The evolution of the state $\rho$ can then be described \cite{GL76,GKS76} by a dynamical semigroup generated by an operator $\mathcal{L}$ of Lindblad form, 
\begin{equation}\label{Lindblad_op}
\mathcal{L}\rho=-i[H,\rho]+\sum_{\alpha\in I}\Gamma_{\alpha}\,\rho\,\Gamma_{\alpha}^*-\frac{1}{2}(\Gamma_{\alpha}^*\Gamma_{\alpha}\,\rho+\rho\,\Gamma_{\alpha}^*\Gamma_{\alpha}),
\end{equation}
see Section \ref{subsection_Lindblad_dynamics}. We consider a special case of such generators, or Lindbladians, the so-called dephasing Lindbladians, characterized by jump operators $\Gamma_\alpha$ of the form $\Gamma_{\alpha}=f_{\alpha}(H)$ for some functions $f_{\alpha}$, see Section \ref{subsubsection_Dephasing_Lindbladians}. Dephasing Lindbladians are of interest in the context of transition probabilities since the space of stationary states is not just $1$-dimensional, as opposed to the generic case. Furthermore we restrict our attention to two-level systems, since Landau-Zener transitions are expected to be dominant and to occur between just two levels at a time. More precisely, we consider the dynamics of a two-level system depending on time through the above parameter $s$. As we will show in Section \ref{subsection_two_level_dephasing_lindbladians}, the most general form of a dephasing two-level Lindbladian is
\begin{equation*}
\mathcal{L}_s\rho(s)=-i[H_s,\rho(s)]-\frac{\gamma_s}{2}[\sqrt{H_s},[\sqrt{H_s},\rho(s)]],
\end{equation*}
where $\gamma_s\geq 0$, $H_s$ is a Hamiltonian acting on $\mathbb{C}^2$, and $\sqrt{H_s}$ is shorthand for $\sgn (H_s)\sqrt{|H_s|}$. We assume $H_s$ and $\gamma_s$ to be smooth in $s$. Moreover, we focus on Hamiltonians $H_s$ going through an avoided crossing (see Figure~\ref{fig_avoided_crossing}) as a function of $s$. Near that point they can be approximatively described by the Landau-Zener Hamiltonian
\begin{equation*}
H_s=\frac{1}{2}\begin{pmatrix}
s &g \\ g& -s
\end{pmatrix}, \quad (g\in\mathbb{R},\, g>0),
\end{equation*}
with eigenvalues $\pm e_{s}=\pm 1/2 \sqrt{s^2+g^2}$ and eigenprojections denoted by $P_s^{\pm}$. 
\begin{figure}[!ht]
\centering
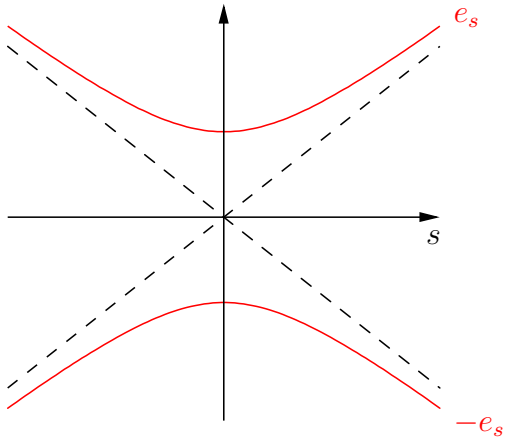
\caption{Energy levels $\pm e_s$ of the Landau-Zener Hamiltonian going through an avoided crossing at $s=0$.}
\label{fig_avoided_crossing}
\end{figure}
The propagator of the Lindblad equation 
\begin{equation*}
\epsilon\dot{\rho}(s)=\mathcal{L}_s\rho(s)
\end{equation*}
is the two-parameter semigroup satisfying 
\begin{gather*}
\epsilon\frac{\partial}{\partial s}U_{\epsilon}(s,s')=\mathcal{L}_sU_{\epsilon}(s,s'),\quad U_{\epsilon}(s',s')=\mathbb{1},\quad (s\geq s').
\end{gather*}
The \emph{\bf{transition probability}} out of the ground state $P_s^{-}$ is
\begin{equation*}
p=\lim_{T\rightarrow\infty}\tr(P_T^{+}\,U_{\epsilon}(T,-T)P_{-T}^{-}).
\end{equation*}
For $\gamma_s\equiv 0$, where the dynamics is Hamiltonian, Landau and Zener \cite{LL32,CZ32} independently proved that 
\begin{equation}\label{LZ_tunneling}
p=\exp\bigl(-\pi\frac{g^2}{2\epsilon}\bigr).
\end{equation}
The transition probability within a finite time interval $(-T,T)$ has also been discussed \cite{MVB90, BT06}. The behaviour is oscillatory in $T$, at least if transitions are defined between instantaneous eigenstates, as done above. A monotone increase in $T$ is obtained if instead of the eigenbasis a modified (superadiabatic) basis is used to define transitions.
 
For general $\gamma_s\geq 0$, and the adiabatic evolution taking place during a finite time interval $[s_0,s_1]$, it was found \cite{AFGG11} that
\begin{equation}\label{LZ_tunneling_afgg}
p(\epsilon,\gamma)=\epsilon \int_{s_0}^{s_1}\,\frac{\gamma_{\tau}}{1+\gamma_{\tau}^2}\frac{\tr(P_{\tau}^{-}(\dot{P}_{\tau}^{+})^2P_{\tau}^{-})}{e_{\tau}}\mathrm{d}\tau+O(\epsilon^2)
\end{equation}
as $\epsilon\to 0$. The leading term, which is of first order in $\epsilon$, vanishes for $\sup_s\gamma_s\rightarrow 0$, yet one does not manifestly recover the Landau-Zener formula (\ref{LZ_tunneling}) in that limit. The expression accounts for the transition probability in terms of transitions between instantaneous eigenstates occurring at specific times within the interval, resulting in a monotone increase at varying rates. Actually, the rates are determined by the velocity of $P_{\tau}^{+}$ in the Fubini-Study metric, which is given here by
\begin{equation}\label{explicit_numerator}
\tr(P_{\tau}^{-}(\dot{P}_{\tau}^{+})^2P_{\tau}^{-})=\frac{g^2}{64\,e_{\tau}^4}.
\end{equation}
The result presented in this paper interpolates between the two results. More precisely, we shall show that for $\gamma_s$, $\dot{\gamma}_s$, and $\ddot{\gamma}_s$ bounded, we have
\begin{equation}\label{theorem_L_Z_generalized_formula}
\begin{aligned}
p(\epsilon,\gamma)=&\exp\bigl(-\pi \frac{g^2}{2\epsilon}\bigr)+\epsilon \int_{-\infty}^{\infty}\,\frac{\gamma_{\tau}}{1+\gamma_{\tau}^2}\frac{\tr(P_{\tau}^{-}(\dot{P}_{\tau}^{+})^2P_{\tau}^{-})}{e_{\tau}}\mathrm{d}\tau+O(\gamma\epsilon^2),
\end{aligned}
\end{equation}
where $\gamma\coloneqq \sup_{\tau}\gamma_{\tau}$. We remark that the asymptotics is uniform in $\epsilon, \gamma$, but not in $g$. In fact, by scaling $p(\epsilon,\gamma)$ actually depends just on $ g^2/\epsilon, \gamma$. It is worth noticing that the expanded part of the general expression (\ref{theorem_L_Z_generalized_formula}) is simply the sum of its two limiting cases (\ref{LZ_tunneling}-\ref{LZ_tunneling_afgg}) for which transitions are a purely coherent and a fully incoherent process, respectively. 

The separation of coherent and incoherent contributions to the tunneling can be realized for any dephasing Lindbladian (\ref{Lindblad_op}), see (\ref{expansion_p}). Our method allows to derive a formula akin to Eq.~(\ref{theorem_L_Z_generalized_formula}) in a more general setting; such extensions are briefly discussed in Section~\ref{sec:extensions}. 


\section{Preliminaries}\label{section_preliminaries}

In this preliminary part we explain the concepts of Linbladians and adiabatic evolution. In particular, the special case of a dephasing Lindbladian is introduced. Moreover, transition probabilities and the adiabatic theorem in terms of dephasing Lindbladians are stated. 
%
%

\begin{subsection}{Lindblad dynamics}\label{subsection_Lindblad_dynamics}
A (super)\footnote{Super-operators are operators acting on bounded operators on the Hilbert space. They will be denoted by calligraphic characters.} operator $\mathcal{L}$ of the form (\ref{Lindblad_op}) with $H=H^*$, $\Gamma_{\alpha}$ arbitrary operators, and $I$ a finite index set, is called a Lindbladian. We write 
\begin{equation}
\mathcal{L}=(H,\Gamma)
\end{equation} 
for short, where $\Gamma$ represents the set of all $\Gamma_{\alpha}$'s. Let us assume that $\dim \mathcal{H}<\infty$, where $\mathcal{H}$ is the Hilbert space associated to the system $\mathcal{S}$. Then if the open system is Markovian, its master equation describing the evolution of a state $\rho(t)$ is of Lindblad form:
\begin{equation}\label{master_eq}
\frac{\mathrm{d}}{\mathrm{d}t}\rho(t)=\mathcal{L}\rho(t).
\end{equation}
We remark that $\mathcal{L}$ is invariant under so-called gauge transformations
\begin{equation}\label{gauge_transformations}
\begin{gathered}
H\mapsto H+e\mathbb{1}-i\sum_{\alpha}(c_{\alpha}^*\Gamma_{\alpha}-c_{\alpha}\Gamma_{\alpha}^*),\quad \Gamma_{\alpha}\mapsto \Gamma_{\alpha}+c_{\alpha}\mathbb{1},\quad (c_{\alpha}\in\mathbb{C}, e\in\mathbb{R}),
\end{gathered}
\end{equation}
as well as
\begin{gather*}
\Gamma\mapsto \mathcal{U}\Gamma, \quad (\mathcal{U}\Gamma)_{\alpha}=\sum_{\beta}\mathcal{U}_{\alpha\beta}\Gamma_{\beta},\quad \mathcal{U}^{-1}=\mathcal{U}^*.
\end{gather*}
Stationary states $\rho$ are elements $\rho\in \ker\mathcal{L}$ by (\ref{master_eq}). The (super) projections on the kernel and the range of $\mathcal{L}$, in the direction of the other, are denoted by $\mathcal{P}$ and $\mathcal{Q}$.
There are several norms which can be associated to operators. We use the same notation for the norm when talking of vectors in a Banach space and of associated bounded operators. In particular we do so for the norm $\Vert \cdot \Vert_1$ (resp. $\Vert \cdot \Vert$) of the space $\mathcal{J}_1(\mathcal{H})$ of trace class operators (resp. $\mathcal{B}(\mathcal{H})$ of bounded operators).
\begin{subsubsection}{Dephasing Lindbladians}\label{subsubsection_Dephasing_Lindbladians}
A Lindbladian $\mathcal{L}=(H,\Gamma)$ is called \emph{\textbf{dephasing}} if
\begin{equation}\label{def_dephasing_Lind}
\Gamma_{\alpha}=f_{\alpha}(H)
\end{equation}
for bounded Borel functions $f_{\alpha}$. This implies $\mathcal{L}P=0$ whenever $[H,P]=0$ and in particular for any spectral projection $P$. Since $H$ is acting on a Hilbert space of (finite) dimension $n$, the stationary states of $\mathcal{L}$ are those of $[H,\cdot]$; in fact the latter statement is in this case equivalent to (\ref{def_dephasing_Lind}), see \cite{AFGG12}. Put differently, the stationary states of $\mathcal{L}$ are the incoherent superpositions of eigenprojections of $H$. The projections $\mathcal{P}$ and $\mathcal{Q}$ can be written as
\begin{equation}\label{projection_sum}
\mathcal{P}\rho=\sum_{j} P^{j}\rho P^{j},\quad \mathcal{Q}\rho=\sum_{j\not=k} P^{j}\rho P^{k},
\end{equation}
where $P^{j}$ are the projections onto the eigenspaces of $H$ \cite{AFG13}. 
If $H$ has simple eigenvalues $e^{0},\dots,e^{n-1}$ with eigenvectors $\psi^i$, the operators $E^{ij}\coloneqq \ket{\psi^i}\bra{\psi^j}$ form a basis of $\mathcal{B}(\mathcal{H})$. In particular, this basis is orthonormal once that space is endowed with the Hilbert-Schmidt inner product.
In the case of a time-dependent operator $H_t$, the above conclusions hold pointwise in $t$.
\end{subsubsection}
\end{subsection}

\begin{subsection}{Adiabatic evolution}\label{subsection_adiabatic_evolution}
The Lindbladian $\mathcal{L}$ may depend on time $t$ through some parameter $s=\epsilon t$, $\epsilon>0$. More precisely, the operators $H$ and $\Gamma_{\alpha}$ may depend on $s$, and thus define a Lindbladian $\mathcal{L}_s$ for each fixed $s$ through (\ref{Lindblad_op})\footnote{The time dependence of objects is denoted by a subscript whenever the dependence is parametric rather than dynamical.}. The master equation (\ref{master_eq}) is accordingly modified to
\begin{equation}\label{adiabatic_master_eq}
\epsilon\frac{\mathrm{d}}{\mathrm{d}s}\rho(s)=\mathcal{L}_s\rho(s).
\end{equation}
We speak about the adiabatic limit when $\epsilon\rightarrow 0$. 
Let now $H_s$ and $\Gamma_{\alpha,s}$ be smooth functions of $s\in \mathbb{R}$. Consequently, $\mathcal{L}_s$ and $\mathcal{L}_s^*$ are smooth as well, $\mathcal{L}_s^*$ being the dual operator with respect to the duality $\mathcal{B}(\mathcal{H})\cong (\mathcal{J}_1(\mathcal{H}))^*$. This is enough to write $\rho(s)=\mathcal{U}_{\epsilon}(s,s')\rho(s')$, with $\mathcal{U}_{\epsilon}(s,s')$ a two-parameter semigroup satisfying the ordinary differential equation
\begin{gather}\label{evolution_eq_propagator}
\epsilon\frac{\partial}{\partial s}\mathcal{U}_{\epsilon}(s,s')=\mathcal{L}_s\mathcal{U}_{\epsilon}(s,s'),\quad \mathcal{U}_{\epsilon}(s',s')=\mathbb{1},\quad (s\geq s').
\end{gather}
We call $\mathcal{U}_{\epsilon}(s,s')$ the propagator corresponding to the Lindbladian $\mathcal{L}_s$. It is a completely positive, trace-preserving (CPTP) map acting on trace class operators, thus $\Vert \mathcal{U}_{\epsilon}(s,s')\Vert_1 =1$ (\cite{BR79}, Cor. 3.6.2).
\begin{subsubsection}{Transition probabilities for dephasing Lindbladians}\label{subsubsection_tun_deph_Lindb}
The projections $P^i$ onto the eigenspaces of $H$ are stationary states of a dephasing Lindbladian $\mathcal{L}$, though not so if it depends on time $s$. We assume these eigenspaces to have constant dimension for all $s$, and that it equals $1$ for the lowest eigenvalue; the corresponding projection is called the ground state and denoted by $P^0$. Then $P_s^0$ is likewise smooth and so is its complementary projection $P_s^{0,\perp}=\mathbb{1}-P_s^0$. The transition probability $p(\epsilon)$ out of the ground state is given by\footnote{Occasionally we may write $p(\epsilon,\boldsymbol{\alpha})$, $\boldsymbol{\alpha}$ a multiindex of quantities of interest in the present situation.}
\begin{equation}
\begin{gathered}
p(\epsilon)=\lim_{T\rightarrow\infty}p(\epsilon,-T,T),\\
p(\epsilon,T,-T)=\tr(P_T^{0,\perp}\,\mathcal{U}_{\epsilon}(T,-T)P_{-T}^{0}).
\end{gathered}
\end{equation}
We also observe that the eigenvalue zero of $\mathcal{L}_s$ has constant degeneracy by the spectral assumptions just made. In particular $\mathcal{L}_s$ has a gap, meaning that the eigenvalue zero is uniformly isolated in $s$, and $\mathcal{P}_s$ is likewise smooth.
\end{subsubsection}
\begin{subsubsection}{Parallel transport}
In order to state the adiabatic theorem, we first recall the concept of parallel transport. As before let $\mathcal{P}_s$ be the projection onto $\ker\mathcal{L}_s$, and observe that it is a CPTP map, as seen from (\ref{projection_sum}) in the dephasing case, but actually true for general Lindbladians \cite{AFGG12}. We recall that $\ran \mathcal{P}_s = \ker\mathcal{L}_s$ is the space of (instantaneous) stationary states at time $s$, and it pays to call the bundle over $s\in \mathbb{R}$ with fiber $\ker\mathcal{L}_s$ the stationary manifold. 
We will show below that to leading order in $\epsilon$, the evolution of the stationary manifold, generated by (\ref{adiabatic_master_eq}), is given by the action of a parallel transport, $\mathcal{T}(s,s')$. Parallel transport is the solution to the evolution equation
\begin{equation}\label{eq_par_tr}
\frac{\partial}{\partial s} \mathcal{T}(s,s') = [\dot{\mathcal{P}}_s, \mathcal{P}_s] \mathcal{T}(s,s'), \quad \mathcal{T}(s',s') = \mathbb{1}.
\end{equation}
The basic properties of parallel transport are recalled in the following proposition, which abstracts from $\ran \mathcal{P}_s = \ker\mathcal{L}_s$ and hence from $\ran \mathcal{P}_s^* = \ker\mathcal{L}_s^*$. 
\begin{proposition}\label{proposition}
A parallel transport has an intertwining property,
\begin{equation}\label{eq_intertwining_property}
\mathcal{T}(s,s') \mathcal{P}_{s'} = \mathcal{P}_s \mathcal{T}(s,s').
\end{equation}
Furthermore, $\mathcal{T}(s,s') \mathcal{P}_{s'} $ is a CPTP map and $\mathcal{T}(s,s')$ maps $\ran \mathcal{P}_{s'}$ isometrically to $\ran \mathcal{P}_s$. Its dual, $\mathcal{T}(s,s')^*$, maps $\ran \mathcal{P}_s^*$ isometrically to $\ran \mathcal{P}_{s'}^*$. 
\end{proposition}
\begin{proof}
The intertwining property follows because both sides satisfy the differential equation (\ref{eq_par_tr}) with the same initial condition. Let us denote the inverse of $\mathcal{T}(s,s')$ by $\mathcal{T}(s',s)$. To see that $\mathcal{T}(s,s')\mathcal{P}_{s'}$ is a CPTP map, first note that $\mathcal{T}(s,s')\mathcal{T}(s',s'')=\mathcal{T}(s,s'')$ by (\ref{eq_par_tr}). Furthermore,
\begin{equation*}
\mathcal{P}_s\mathcal{P}_{s'}+\mathcal{Q}_s\mathcal{Q}_{s'}=\mathbb{1}+[\dot{\mathcal{P}}_{s'},\mathcal{P}_{s'}](s-s')+o(|s-s'|),\quad (s\rightarrow s')
\end{equation*}
by Taylor expansion. Here $\mathcal{Q}_s=\mathbb{1}-\mathcal{P}_s$. Thus we can write 
\begin{align}
\mathcal{T}(s,s')=&\lim_{N\rightarrow \infty}\prod_{i=0}^{N-1}(\mathcal{P}_{s_{i+1}}\mathcal{P}_{s_i}+\mathcal{Q}_{s_{i+1}}\mathcal{Q}_{s_i})\nonumber\\
=&\lim_{N\rightarrow \infty}(\prod_{i=0}^{N}\mathcal{P}_{s_{i}}+\prod_{i=0}^{N}\mathcal{Q}_{s_i}),\label{par_tr_as_product}
\end{align}
where $s'=s_0\leq s_1\leq \dots \leq s_N=s$ is a partition of $[s',s]$ into intervals of length $|s_{i+1}-s_i|=N^{-1}|s-s'|$ and $\prod_{i=0}^N A_i\coloneqq A_N\dots A_0$.
From (\ref{par_tr_as_product}) we see that $\mathcal{T}(s,s') \mathcal{P}_{s'} $ indeed is a CPTP map. Thus $\Vert \mathcal{T}(s,s') \mathcal{P}_{s'}\Vert_1 =1$ (\cite{BR79}, Cor. 3.6.2).
Consequently, for $\rho \in \ran\mathcal{P}_{s'}$,
\begin{equation*}
\Vert \mathcal{T}(s,s')\rho \Vert_1\leq \Vert \rho \Vert_1 =\Vert \mathcal{T}(s',s)\mathcal{T}(s,s')\rho \Vert_1\leq \Vert \mathcal{T}(s,s')\rho \Vert_1,
\end{equation*}
establishing the isometry property. The properties of the dual map follow from 
\begin{equation}\label{dual_par_tr}
\mathcal{T}(s,s')^*=\mathcal{T}^*(s',s),
\end{equation}
where $\mathcal{T}^*(s,s')$ is the parallel transport corresponding to the dual projection $\mathcal{P}_s^*$. Equation (\ref{dual_par_tr}) in turn follows from (\ref{eq_par_tr}) and 
\begin{equation*}
\frac{\partial}{\partial s} \mathcal{T}(s',s) =-  \mathcal{T}(s',s)[\dot{\mathcal{P}}_s, \mathcal{P}_s].
\end{equation*}
\end{proof}
In the case of $\mathcal{P}_s$ corresponding to some Lindbladian $\mathcal{L}_s$, note that if the latter is dephasing, then $\mathcal{P}_s =\mathcal{P}_s^*$ as operators on $\mathcal{B}(\mathcal{H})$, and in particular $\ran \mathcal{P}_s = \ran \mathcal{P}_s^*=\ker[H,\cdot]$.
\end{subsubsection}
\begin{subsubsection}{The adiabatic theorem for Lindbladians}
Let the family $P_s\in \ker\mathcal{L}_s$ be compatible with parallel transport,
\begin{equation}\label{par_tr_kernel}
\mathcal{T}(s,s') P_{s'} = P_{s}
\end{equation}
in line with Proposition \ref{proposition}. We then call $P_s$ a family of parallel transported stationary states. If $\mathcal{L}_s$ is dephasing, the parallel transport $\mathcal{T}$ is associated to $\mathcal{P}_s$ seen in (\ref{projection_sum}), and the spectral projections $P_s=P_s^j$ of $H_s$ are examples of such families. In fact both sides of (\ref{par_tr_kernel}) then satisfy the same differential equation (\ref{eq_par_tr}). Thus also fixed convex combinations of the $P_s^j$ are families of parallel transported stationary states. However, such families may also occur for general Lindbladians $\mathcal{L}_s$ which have a gap, i.e.\ for which the eigenvalue zero is uniformly isolated in $s$, and thus the projection $\mathcal{P}_s$ onto $\ker\mathcal{L}_s$ is likewise smooth. In fact $P_s\coloneqq \mathcal{T}(s,s')P_{s'}\in\ker\mathcal{L}_s$ is such a family by construction for any $P_{s'}\in\ker\mathcal{L}_{s'}$; in case the spaces are $1$-dimensional, the states $P_s$ are unique.
Based on the notion of parallel transport we can now formulate the adiabatic theorem. It provides solutions of (\ref{adiabatic_master_eq}) that remain close to the stationary manifold to first order in $\epsilon$. The statement here is a special case of Theorem~6 in \cite{AFGG12} which gives an expansion to an arbitrary order. For completeness we give a proof of the version stated here. 
\begin{theorem}[Adiabatic Theorem]\label{adiabatic_theorem}
Let $\mathcal{L}_s$ be a Lindbladian with a gap, $\mathcal{P}_s$ the projection onto $\ker\mathcal{L}_s$, and $P_s$ a family of parallel transported stationary states of $\mathcal{L}_s$. Then the driven Lindblad equation $\epsilon \dot{\rho}(s) = \mathcal{L}_s \rho(s)$ admits a solution of the form
\begin{equation}\label{solution_ad_th}
\begin{gathered}
\rho(s) = P_s + \epsilon\, a_{s,s'} + \epsilon^2 r(s,s'),\\
a_{s,s'}\coloneqq \mathcal{L}^{-1}_s \dot{P}_s + \int_{s'}^s \mathcal{T}(s,\tau)\dot{\mathcal{P}}_{\tau} \mathcal{L}_{\tau}^{-1} \dot{P}_\tau \,\mathrm{d} \tau,\quad (s\geq s').
\end{gathered}
\end{equation}
The remainder term $r(s,s')$ has an expression 
\begin{equation}\label{rest_term_ad_th}
r(s,s')\coloneqq b_{s,s'}-\mathcal{U}_{\epsilon}(s,s')b_{s',s'}-\int_{s'}^s \mathcal{U}_{\epsilon}(s,\tau)\dot{b}_{\tau,s'}\,\mathrm{d}\tau,
\end{equation}
\begin{equation}\label{def_b_j}
b_{s,s'}\coloneqq \mathcal{L}_s^{-1}\dot{\mathcal{P}}_s\int_{s'}^s \mathcal{T}(s,\tau)\dot{\mathcal{P}}_{\tau}\mathcal{L}_{\tau}^{-1}\dot{P}_{\tau}\,\mathrm{d}\tau+\mathcal{L}_s^{-1}(\mathbb{1}-\mathcal{P}_s)\frac{\mathrm{d}}{\mathrm{d}s}(\mathcal{L}_s^{-1}\dot{P}_s),
\end{equation}
and is uniformly bounded in $\epsilon$ for $s',s$ finite.
\end{theorem}
Here $\mathcal{L}^{-1}_s$ denotes the inverse of the operator $\mathcal{L}_s$, which is defined on $\ran\mathcal{L}_s$. It is moreover bounded because of the gap of $\mathcal{L}_s$. By (\ref{eq_par_tr}, \ref{par_tr_kernel}) and $\mathcal{P}_s[\dot{\mathcal{P}}_s,\mathcal{P}_s]P_s=0$ we have $\mathcal{P}_s\dot{P}_s=0$, and hence $\dot{P}_s\in \ran\mathcal{L}_s$. The above expressions involving $\mathcal{L}^{-1}_s$ are thus well-defined. They are so also w.r.t.\ the derivatives, since $\mathcal{P}_{s}$, and hence $P_s$ and $\mathcal{L}_{s}^{-1}$, are smooth like $\mathcal{L}_{s}$ by the gap condition.
\begin{proof}
The uniform boundedness in $\epsilon$ of $r(s,s')$ for $s',s$ finite follows from $\Vert \mathcal{U}_{\epsilon}(s,s') \Vert_1= 1$, $(s'\leq s)$.
$P_s$ is mapped to $0$ under $\mathcal{L}_s$, and the same holds for the integral in (\ref{solution_ad_th}). The latter can easily be seen by remarking 
\begin{equation}\label{eq_p_l_0}
\mathcal{P}_{\tau} \mathcal{L}_{\tau}^{-1} \dot{P}_\tau=0,
\end{equation}
and hence 
\begin{equation*}
\mathcal{T}(s,\tau)\dot{\mathcal{P}}_{\tau} \mathcal{L}_{\tau}^{-1} \dot{P}_\tau=-\mathcal{T}(s,\tau)\mathcal{P}_{\tau} \frac{\mathrm{d}}{\mathrm{d}\tau}(\mathcal{L}_{\tau}^{-1} \dot{P}_\tau)\in \ker\mathcal{L}_s.
\end{equation*}
Thus for $\rho(s)$ given by (\ref{solution_ad_th}) we obtain
\begin{align*}
\mathcal{L}_s \rho(s)=&\epsilon \dot{P}_s+ \epsilon^2 \Bigl(\dot{\mathcal{P}}_s\int_{s'}^s \mathcal{T}(s,\tau)\dot{\mathcal{P}}_{\tau}\mathcal{L}_{\tau}^{-1}\dot{P}_{\tau}\,\mathrm{d}\tau+(\mathbb{1}-\mathcal{P}_s)\frac{\mathrm{d}}{\mathrm{d}s}(\mathcal{L}_s^{-1}\dot{P}_s)\Bigr)\\
&-\epsilon^3 \Bigl(\frac{\partial}{\partial s} (\mathcal{U}_{\epsilon}(s,s'))b_{s',s'}+\int_{s'}^s \frac{\partial}{\partial s}(\mathcal{U}_{\epsilon}(s,\tau))\dot{b}_{\tau,s'}\,\mathrm{d}\tau\Bigr),
\end{align*}
where we used (\ref{evolution_eq_propagator}) for the last two terms in (\ref{rest_term_ad_th}). Regarding the terms which are linear or cubed in $\epsilon$, this formula obviously matches the corresponding terms in $\epsilon\dot{\rho}(s)$. For those squared in $\epsilon$, we use that by (\ref{eq_p_l_0})
\begin{equation*}
(\mathbb{1}-\mathcal{P}_s)\frac{\mathrm{d}}{\mathrm{d}s}(\mathcal{L}_s^{-1}\dot{P}_s)=\frac{\mathrm{d}}{\mathrm{d}s}(\mathcal{L}_s^{-1}\dot{P}_s)+\dot{\mathcal{P}}_s\mathcal{L}_s^{-1}\dot{P}_s,
\end{equation*}
and by (\ref{eq_par_tr}), $\dot{\mathcal{P}}_s=\dot{\mathcal{P}}_s\mathcal{P}_s+\mathcal{P}_s\dot{\mathcal{P}}_s$, (\ref{eq_intertwining_property}) and (\ref{eq_p_l_0})
\begin{align*}
\frac{\partial}{\partial s}(\mathcal{T}(s,\tau))\dot{\mathcal{P}}_{\tau}\mathcal{L}_{\tau}^{-1} \dot{P}_\tau=&[\dot{\mathcal{P}}_s,\mathcal{P}_s]\mathcal{T}(s,\tau)\dot{\mathcal{P}}_{\tau}\mathcal{L}_{\tau}^{-1} \dot{P}_\tau\\
=&(2\dot{\mathcal{P}}_s\mathcal{P}_s-\dot{\mathcal{P}}_s)\mathcal{T}(s,\tau)\dot{\mathcal{P}}_{\tau}\mathcal{L}_{\tau}^{-1} \dot{P}_\tau\\
=&\dot{\mathcal{P}}_s\mathcal{T}(s,\tau)(2\mathcal{P}_{\tau}\dot{\mathcal{P}}_{\tau}-\dot{\mathcal{P}}_{\tau})\mathcal{L}_{\tau}^{-1} \dot{P}_\tau\\
=&\dot{\mathcal{P}}_s\mathcal{T}(s,\tau)\dot{\mathcal{P}}_{\tau}\mathcal{L}_{\tau}^{-1} \dot{P}_\tau.
\end{align*} 
\end{proof}
\end{subsubsection}
\end{subsection}


\section{Two-level dephasing Lindbladians}\label{section_tunneling_for_dephasing_lindbladians}
The main result (Theorem~\ref{theorem_L_Z_generalized} below) computes the transition probability from the ground state for a dephasing Lindbladian corresponding to the Landau-Zener Hamiltonian. Prior to this, we will make general considerations about two-level dephasing Lindbladians. 

\begin{subsection}{Minimally degenerate dephasing Lindbladians}\label{subsection_two_level_dephasing_lindbladians}
A dephasing Lindbladian, recall Section \ref{subsubsection_Dephasing_Lindbladians}, has a kernel of dimension at least $n$, where $n$ is the (finite) dimension of the Hilbert space. Eigenvalues other than $0$ may be non-degenerate. We say that the dephasing Lindbladian is \textbf{\emph{minimally degenerate}} if each eigenvalue is minimally degenerate.

Let us consider such a Lindbladian $\mathcal{L}$ governing the evolution of a two-level system. The following lemma holds true:
\begin{lemma}\label{lemma_generality_Lindbladian}
Any minimally degenerate dephasing Lindbladian $\mathcal{L}=(\tilde{H},\Gamma_{\alpha})$ of a two-level system is of the form
\begin{equation}\label{general_form_two_state_Lindbladian}
\mathcal{L} =\mathcal{L}_{0} +\frac{\gamma}{2}\mathcal{D}.
\end{equation}
Here
\begin{equation}\label{definition_L_0}
\mathcal{L}_{0} \rho(s)=-i[H,\rho(s)],
\end{equation}
\begin{equation}\label{definition_D}
\mathcal{D}\rho(s)=-[\sqrt{H},[\sqrt{H},\rho(s)]],
\end{equation}
and $\gamma\geq 0$ is a real number. Furthermore $H$ is a non-degenerate, traceless Hamiltonian, and $\sqrt{H}\coloneqq\sgn (H)\sqrt{|H|}$.

Moreover, if $\mathcal{L}=\mathcal{L}_s $ depends smoothly on $s$, then so do $H_s$ and $\gamma_s$.
\end{lemma}
\begin{proof}
If $\tilde{H}$ were degenerate, then $\tilde{H}\propto 1$ and $\Gamma_{\alpha}\propto 1$, whence $\mathcal{L}=0$. This possibility is ruled out, because then $\mathcal{L}$ is not minimally degenerate. Let thus $\tilde{H}$ have distinct eigenvalues $e^{+},e^{-}$.
$\mathcal{L}$ is uniquely determined by its action on a basis of $\mathcal{B}(\mathcal{H})$. Let us denote the right hand side of (\ref{general_form_two_state_Lindbladian}) by $\tilde{\mathcal{L}}$. It is enough to show that we find a Hamiltonian $H$, such that $\tilde{\mathcal{L}}$ has the same action on the basis elements as $\mathcal{L}$. Let $P^{\pm}\equiv \ket{\psi^{\pm}}\bra{\psi^{\pm}}$ be the eigenprojections corresponding to the eigenvalues $e^{\pm}$, and define $E=\ket{\psi^{+}}\bra{\psi^{-}}$. Note that $E$ is only defined up to a phase; however, for this proof the choice of this phase is not relevant.
The set $\{P^{+},P^{-},E,E^*\}$ constitutes a basis of $\mathcal{B}(\mathcal{H})$. By the definition of a dephasing Lindbladian, we have $\Gamma_{\alpha}=f_{\alpha}(\tilde{H})$ and thus
\begin{gather*}
\mathcal{L}(P^{\pm})=0,\quad \mathcal{L}(E)=\lambda E,\quad \mathcal{L}(E^*)=\overline{\lambda} E^*,
\end{gather*}
with
\begin{gather}\label{eigenv_lambda}
\lambda=-i(e^{+}-e^{-})+\sum_{\alpha\in I}f_{\alpha}^{+}\overline{f_{\alpha}^{-}}-\frac{1}{2}(\overline{f_{\alpha}^{+}}f_{\alpha}^{+}+\overline{f_{\alpha}^{-}}f_{\alpha}^{-}),\quad
f_{\alpha}^{\pm}=f_{\alpha}(e^{\pm}),
\end{gather}
and 
\begin{equation*}
\Re(\lambda)=-\frac{1}{2}\sum_{\alpha\in I}|f_{\alpha}^{+}-f_{\alpha}^{-}|^2\leq 0.
\end{equation*}
Let us set $H=\kappa(P^{+}-P^{-})$, with $\kappa$ to be determined. Then 
\begin{gather*}
\tilde{\mathcal{L}}(P^{\pm})=0,\quad \tilde{\mathcal{L}}(E)=(- 2i\kappa-2\gamma|\kappa|)E,\quad \tilde{\mathcal{L}}(E^*)=( 2i\kappa-2\gamma|\kappa|)E^*,
\end{gather*}
resulting in 
\begin{equation*}
\lambda,\overline{\lambda}=\mp2i\kappa-2\gamma|\kappa|.
\end{equation*}
We have $\Im(\lambda)\not=0$ because $\mathcal{L}$ is minimally degenerate. The equation can thus be solved for $\kappa,\gamma$, yielding
\begin{gather*}
\kappa=\frac{1}{2}\left(e^{+}-e^{-}-\sum_{\alpha\in I}\Im(f_{\alpha}^{+}\overline{f_{\alpha}^{-}})\right)\,,\quad
\gamma=\frac{1}{4|\kappa|}\sum_{\alpha\in I}|f_{\alpha}^{+}-f_{\alpha}^{-}|^2\,.
\end{gather*}

If $\mathcal{L}_s=(\tilde{H}_s,\Gamma_{\alpha,s})$ are smooth in $s$, then so are $P^{\pm}_s$ and, by inspection, $\kappa_s,\gamma_s$.
\end{proof}
We consider a family of two-level non-degenerate Hamiltonians near a minimum of the gap. By choosing an appropriate basis and parametrization the behaviour is captured by
\begin{equation}\label{LZ_Hamiltonian}
H_s=\frac{1}{2}\begin{pmatrix}
s & g \\ g& -s
\end{pmatrix}=e_s(P_s^{+}-P_s^{-}),
\end{equation}
where $g>0$, and $\pm e_{s}=\pm 1/2\sqrt{s^2+g^2}$ and $P_s^{\pm}$ are the eigenvalues and eigenprojections of $H_s$ respectively.
\begin{lemma}\label{ad_th_for_two_level_Lindbladian}
Let $\mathcal{L}_s$ be given by (\ref{general_form_two_state_Lindbladian}-\ref{definition_D}) and (\ref{LZ_Hamiltonian}). Then $\epsilon \dot{\rho}(s) = \mathcal{L}_s \rho(s)$ admits solutions
\begin{equation}\label{solution_ad_th_2}
\rho^{\pm}(s) = P^{\pm}_s + \epsilon \,a^{\pm}_{s,s'} + \epsilon^2 r^{\pm}(s,s'),
\end{equation}
with
\begin{equation}\label{def_a_pm}
a^{\pm}_{s,s'}\coloneqq \pm\frac{g((i-\gamma_s)E_s+h.c.)}{16(1+\gamma_s^2)e_s^3}\pm \frac{g^2(P_s^{-}-P_s^{+})}{64}\int_{s'}^s \frac{\gamma_{\tau}}{(1+\gamma_{\tau}^2)e_{\tau}^5}\,\mathrm{d} \tau,
\end{equation}
$(s\geq s')$, and a specific choice of $E_s$ given below. Moreover, if $\gamma_s$ as well as its first and second derivative are bounded continuous functions, then $r^{\pm}(s,s')$ is uniformly bounded in $\epsilon$ and $s,s'$.
\end{lemma}
The dephasing property of $\mathcal{L}_s$ implies $\rho^{+}(s)+\rho^{-}(s)=\mathbb{1}$, which is reflected in the expansion through $a^{-}_{s,s'} = - a^{+}_{s,s'}$.
Note that the last part of the statement is a strengthening with respect to Theorem~\ref{adiabatic_theorem}, since there the remainder term is uniformly bounded only for bounded $s,s'$.

We will later use the lemma for $\mathcal{L}_{-s}^*$ instead of $\mathcal{L}_s$. By the remark preceding Proposition \ref{proposition}, $\mathcal{T}(-t,-\sigma)$ is identical for $\mathcal{L}_{-t}^*$ and $\mathcal{L}_{-t}$. Moreover, $\mathcal{L}_{-t}^*$ has the same eigenstates as $\mathcal{L}_{-t}$, though the corresponding eigenvalues are complex conjugated. This implies the following adaption of $a^{\pm}_{s,s'}$:
\begin{equation}\label{def_hat_a_pm}
\hat{a}^{\pm}_{-s',-s}\coloneqq \mp\frac{g((i-\gamma_{-s})E_{-s}^*+h.c.)}{16(1+\gamma_{-s}^2)e_{-s}^3}\pm \frac{g^2(P_{-s}^{-}-P_{-s}^{+})}{64}\int_{s'}^s \frac{\gamma_{-\tau}}{(1+\gamma_{-\tau}^2)e_{-\tau}^5}\,\mathrm{d} \tau\,.
\end{equation}

As mentioned in the proof of Lemma \ref{lemma_generality_Lindbladian}, $E_{s}$ is only defined up to a phase. We  choose it to be real, i.e.\
\begin{equation}\label{E+-}
E_{s}=\frac{1}{4e_{s}}\begin{pmatrix}
g & -s- 2e_s\\
-s+2e_s & -g
\end{pmatrix}\,.
\end{equation}
\begin{proof}
By Theorem~\ref{adiabatic_theorem},
\begin{gather*}
a^{\pm}_{s,s'}=\mathcal{L}^{-1}_s \dot{P}^{\pm}_s + \int_{s'}^s \mathcal{T}(s,\tau)\dot{\mathcal{P}}_{\tau} \mathcal{L}_{\tau}^{-1} \dot{P}^{\pm}_\tau \,\mathrm{d} \tau. 
\end{gather*}
We compute
\begin{equation}\label{P_dot_lin_comb}
\dot{P}_{s}^{\pm}=\pm \frac{g}{8\,e_s^2}(E_{s}+E_{s}^*)
\end{equation}
and
\begin{gather}\label{action_inverse_super_op}
\mathcal{L}_{s}^{-1}E_{s}=\frac{E_{s}}{2(- i-\gamma_{s})e_s},\quad \mathcal{L}_{s}^{-1}E_{s}^*=(\mathcal{L}_{s}^{-1}E_{s})^*,
\end{gather}
which together determine $\mathcal{L}_s^{-1}\dot{P}_s^{\pm}$ and yield the first term of (\ref{def_a_pm}). The second term is then seen to have the stated form by the use of
\begin{equation}\label{action_dot_P}
\dot{\mathcal{P}}_{s}E_{s}=\dot{\mathcal{P}}_{s}E_{s}^*=\frac{g}{8\,e_s^2}(P_s^{+}-P_s^{-})
\end{equation}
and (\ref{par_tr_kernel}).
Finally, the last term $r^{\pm}(s,s')$ was computed in (\ref{rest_term_ad_th}). There
\begin{equation}\label{norm_U}
\Vert \mathcal{U}_{\epsilon}(s,s')\Vert_1=1,
\end{equation}
like before, and $\Vert b^{\pm}_{s,s'}\Vert_1=O(s^{-3})$, which is manifest from the exact calculation of $b^{\pm}_{s,s'}$, see Appendix \ref{app:calc:b}. Moreover, exact calculation (Appendix \ref{app:calc:b:dot}) provides
\begin{equation*}
\Vert \dot{b}^{\pm}_{\tau,s'}\Vert_1\leq c\,e_{\tau}^{-3}
\end{equation*}
with $c>0$, and hence the last term in $r^{\pm}(s,s')$ is bounded by a constant for all $s,s'$. 
\end{proof}
\end{subsection}


\begin{subsection}{Transition probabilities for minimally degenerate two-level dephasing Lindbladians}\label{subsection_transition_prob_2_level_dephasing_Lindbladians}
Let $\mathcal{L}_s$ be given by (\ref{general_form_two_state_Lindbladian}-\ref{definition_D}) and (\ref{LZ_Hamiltonian}). The following theorem establishes the transition probability claimed in (\ref{theorem_L_Z_generalized_formula}).
\begin{theorem}\label{theorem_L_Z_generalized}
Let $\gamma_s$ as well as its first two derivatives be bounded
continuous functions, and set $\gamma\coloneqq \sup_s\gamma_s$. 
Then
\begin{equation*}
p(\epsilon,\gamma)=\exp(-\pi \frac{g^2}{2\epsilon})+\epsilon \int_{-\infty}^{\infty}\,\frac{\gamma_{\tau}}{1+\gamma_{\tau}^2}\frac{\tr(P_{\tau}^{-}(\dot{P}_{\tau}^{+})^2P_{\tau}^{-})}{e_{\tau}}\,\mathrm{d}\tau+R,
\end{equation*}
where the remainder $R$ satisfies $|R|\le C\gamma\epsilon^2$. The
constant $C$ depends only on $g$ and a bound on the stated derivatives. 
\end{theorem}
\begin{proof}
Let $\mathcal{U}_{\epsilon}(s,s')$, $(s\geq s')$, be the two-parameter group solving (\ref{evolution_eq_propagator}). Applying Duhamel's formula to (\ref{general_form_two_state_Lindbladian}), we can write
\begin{equation}\label{duhamel_formula}
\mathcal{U}_{\epsilon}(s,s')=\mathcal{U}_{\epsilon,0}(s,s')+\frac{1}{\epsilon}\int_{s'}^s \mathcal{U}_{\epsilon}(s,\tau)\,\frac{\gamma_{\tau}}{2}\,\mathcal{D}_{\tau}\,\mathcal{U}_{\epsilon,0}(\tau,s')\,\mathrm{d}\tau,
\end{equation}
where $\mathcal{U}_{\epsilon,0}(s,s')$ solves the equation
\begin{equation}\label{deq_U_0}
\epsilon\,\frac{\partial}{\partial s}\mathcal{U}_{\epsilon,0}(s,s')=\mathcal{L}_{0,s}\,\mathcal{U}_{\epsilon,0}(s,s').
\end{equation}
It can easily be seen by differentiation that the r.h.s.\ of (\ref{duhamel_formula}) satisfies Equation (\ref{evolution_eq_propagator}). Hence 
\begin{align}
p(\epsilon,\gamma,T)=&\tr(P_T^{+}\,\mathcal{U}_{\epsilon}(T,-T)P_{-T}^{-}) \nonumber \\
=&\tr(P_T^{+}\,\mathcal{U}_{\epsilon,0}(T,-T)P_{-T}^{-}) \nonumber \\
&+\frac{1}{2\epsilon}\int_{-T}^T  \gamma_{\tau}\tr(P_T^{+}\, \mathcal{U}_{\epsilon}(T,\tau)\,\mathcal{D}_{\tau}\,\mathcal{U}_{\epsilon,0}(\tau,-T)P_{-T}^{-})\,\mathrm{d}\tau \nonumber \\
=&\tr(P_T^{+}\,\mathcal{U}_{\epsilon,0}(T,-T)P_{-T}^{-}) \nonumber \\
&+\frac{1}{2\epsilon}\int_{-T}^T  \gamma_{\tau}\tr((\mathcal{U}_{\epsilon}^*(T,\tau)\,P_T^{+})\, \mathcal{D}_{\tau}\,\mathcal{U}_{\epsilon,0}(\tau,-T)P_{-T}^{-})\,\mathrm{d}\tau,\label{expansion_p}
\end{align}
where $\mathcal{U}_{\epsilon}^*$ is the dual operator with respect to the duality $\mathcal{B}(\mathcal{H})\cong (\mathcal{J}_1(\mathcal{H}))^*$. Note that the first term corresponds to the Landau-Zener formula, and hence provides the first term of (\ref{theorem_L_Z_generalized_formula}). Furthermore, note that the (formal) separation of this contribution corresponding to the Hamiltonian evolution in (\ref{expansion_p}) neither depends on the fact of treating a two-level system, nor on the specific form of $\mathcal{D}$; it is thus possible for any Lindbladian (\ref{Lindblad_op}).

We are left to show that the remainder of (\ref{expansion_p}), henceforth denoted by $p_{d}(\epsilon,\gamma,T)$, corresponds to the last two terms of (\ref{theorem_L_Z_generalized_formula}) in the limit $T\rightarrow \infty$. By Lemma \ref{ad_th_for_two_level_Lindbladian}, we may expand $\mathcal{U}_{\epsilon,0}(\tau,-T)P_{-T}^{-}$ as 
\begin{equation}\label{expansion_U_0}
\begin{aligned}
\mathcal{U}_{\epsilon,0}(\tau,-T)P_{-T}^{-}=P_{\tau}^{-} +\epsilon\, a^{-}_{\tau,-T,0}-\epsilon\,\mathcal{U}_{\epsilon,0}(\tau,-T)\, a^{-}_{-T,-T,0}+ \epsilon^2 r^{-}_{0}( \tau,-T),
\end{aligned}
\end{equation}
\begin{equation}\label{def_a_-_0}
a^{-}_{\tau,-T,0}\coloneqq \frac{g(iE^*_{\tau}+h.c.)}{16\,e_{\tau}^3},
\end{equation}
with $r^{-}_{0}( \tau,-T)$ computed as in (\ref{rest_term_ad_th}). More precisely, Lemma \ref{ad_th_for_two_level_Lindbladian} provides a solution $\rho^{-}(\tau)$ of (\ref{deq_U_0}) by setting $\gamma_{\tau}\equiv 0$ in (\ref{solution_ad_th_2}), (\ref{def_a_pm}), i.e.\
\begin{equation*}
\rho^{-}(\tau)=P_{\tau}^{-} +\epsilon\, a^{-}_{\tau,-T,0}+\epsilon^2 r^{-}_{0}( \tau,-T).
\end{equation*}
The expansion then follows by writing 
\begin{equation*}
\rho^{-}(\tau)=\mathcal{U}_{\epsilon,0}(\tau,-T)\rho^{-}(-T)=\mathcal{U}_{\epsilon,0}(\tau,-T)P_{-T}^{-}+\epsilon\,\mathcal{U}_{\epsilon,0}(\tau,-T)\, a^{-}_{-T,-T,0}.
\end{equation*}
The propagator $\mathcal{U}_{\epsilon}^*$ in turn satisfies
\begin{equation}\label{deq_U_star}
\epsilon\, \frac{\partial}{\partial s'}\mathcal{U}_{\epsilon}^*(s,s')=-\mathcal{L}_{s'}^*\,\mathcal{U}_{\epsilon}^*(s,s'),\quad (s\geq s'),
\end{equation}
which can be restated in terms of $\mathcal{U}_{\epsilon}^*(s,s')\eqqcolon \mathcal{V}(-s',-s)$ and $-s'\eqqcolon t$, $-s\eqqcolon t'$, i.e.\
\begin{equation}\label{deq_V}
\epsilon\, \frac{\partial}{\partial t}\mathcal{V}(t,t')=\mathcal{L}_{-t}^*\,\mathcal{V}(t,t'),\quad (t\geq t').
\end{equation}
Thus Theorem~\ref{adiabatic_theorem} applies. By Lemma \ref{ad_th_for_two_level_Lindbladian}, and more precisely (\ref{def_hat_a_pm}), the solution provided by Theorem~\ref{adiabatic_theorem} yields the expansion 
\begin{equation}\label{expansion_U_star}
\mathcal{U}_{\epsilon}^*(T,\tau)P_T^{+}=P_{\tau}^{+} +\, \epsilon\, \hat{a}^{+}_{T,\tau}-\,\epsilon\,\mathcal{U}_{\epsilon}^*(T,\tau)\hat{a}^{+}_{T,T}+ \epsilon^2 \tilde{r}^{+}(T,\tau),
\end{equation}
where 
\begin{equation}
\begin{gathered}
\hat{a}^{+}_{T,\tau}= \tilde{a}^{+}_{T,\tau}+ \frac{ g^2(P_{\tau}^{+}-P_{\tau}^{-})}{64}\int_{T}^{\tau} \frac{\gamma_{\zeta}}{(1+\gamma_{\zeta}^2)e_{\zeta}^5}\,\mathrm{d} \zeta,\\
\tilde{a}^{+}_{T,\tau}\coloneqq  -\frac{ g ((i-\gamma_{\tau})E^*_{\tau}+h.c.)}{16(1+\gamma_{\tau}^2)e_{\tau}^3}
\end{gathered}
\end{equation}
and $\tilde{r}^{+}(T,\tau)$ defined in (\ref{rest_term_ad_th}), with $\tilde{b}_{T,\tau}$ defined as in (\ref{def_b_j}). Note that by Appendices \ref{app:calc:b}-\ref{app:calc:b:dot}, $\Vert \tilde{b}^{+}_{s',s}\Vert_1=O(s^{-3})$ and $\Vert \dot{\tilde{b}}^{+}_{s',s}\Vert_1\leq c\,e_{s}^{-3}$, $c>0$, and hence $\tilde{r}^{+}(s',s)$ is uniformly bounded in $s,s'$.

We insert (\ref{expansion_U_0}, \ref{expansion_U_star}) into (\ref{expansion_p}) and observe that $P_{\tau}^{\pm}$ is in the kernel of $\mathcal{D}_{\tau}=\mathcal{D}_{\tau}^*$. We deduce 
\begin{equation}\label{expansion_p_d}
\begin{aligned}
p_{d}(\epsilon,\gamma,T)=&\epsilon\int_{-T}^T  \frac{\gamma_{\tau}}{2}\tr([\tilde{a}_{T,\tau}^{+}-\mathcal{U}_{\epsilon}^*(T,\tau)\,\tilde{a}_{T,T}^{+}+\epsilon\,\tilde{r}^{+}(T,\tau)]\,\\
&\times \mathcal{D}_{\tau}\,[a_{\tau,-T,0}^{-}-\mathcal{U}_{\epsilon,0}(\tau,-T)\,a_{-T,-T,0}^{-}+\epsilon\,r^{-}_0(\tau,-T)])\,\mathrm{d}\tau.
\end{aligned}
\end{equation}

Expanding the expression inside the trace and using the linearity of the latter we can write
\begin{equation*}
p_{d}(\epsilon,\gamma,T)=\epsilon \int_{-T}^T \frac{\gamma_{\tau}}{2}\,\sum_{i,j=1}^3 \tr(T_{\tau}^{ij})\,\,\mathrm{d}\tau,
\end{equation*}
where $T_{\tau}^{ij}$ denotes the combination of the $i$-th term of the first square bracket and the $j$-th term of the second. 

Plugging in definitions and using 
\begin{equation}\label{D_on_E}
\mathcal{D}_{\tau}E_{\tau}=-4\,e_{\tau}E_{\tau}
\end{equation}
one easily finds 
\begin{equation*}\label{eq_number_1}
\epsilon\int_{-T}^T\frac{\gamma_{\tau}}{2}\tr(T_{\tau}^{11})\,\mathrm{d}\tau
=\frac{\epsilon g^2}{64}\int_{-T}^T\frac{\gamma_{\tau}}{(1+\gamma_{\tau}^2)e_{\tau}^5}\mathrm{d}\tau.
\end{equation*}
Moreover, we have (\ref{explicit_numerator}) by (\ref{P_dot_lin_comb}) and therefore
\begin{equation}\label{eq_number}
\epsilon\int_{-T}^T\frac{\gamma_{\tau}}{2}\tr(T_{\tau}^{11})\,\mathrm{d}\tau
=\epsilon\int_{-T}^T\frac{\gamma_{\tau}}{1+\gamma_{\tau}^2}\frac{\tr(P_{\tau}^{-}(\dot{P}_{\tau}^{+})^2P_{\tau}^{-})}{e_{\tau}}\mathrm{d}\tau,
\end{equation}
which yields the term of first order in $\epsilon$ of (\ref{theorem_L_Z_generalized_formula}). 

Furthermore, by the lemmas below, the integrals corresponding to the $T_{\tau}^{ij}$'s left over either vanish in the limit $T\rightarrow \infty$ (Lemma \ref{lemma_vanishing_terms}), or contribute terms which are $ O(\gamma\epsilon^2)$ (Lemmas \ref{lemma_rest_terms_1}, \ref{lemma_rest_terms_2}). 
\end{proof}
\begin{lemma}\label{lemma_vanishing_terms}
For all $j\in\{1,2,3\}$,
\begin{equation*}
\lim_{T\rightarrow\infty}\int_{-T}^T\frac{\gamma_{\tau}}{2}\tr(T_{\tau}^{2j})\,\mathrm{d}\tau=\lim_{T\rightarrow\infty}\int_{-T}^T\frac{\gamma_{\tau}}{2}\tr(T_{\tau}^{j2})\,\mathrm{d}\tau=0\,.
\end{equation*}
\end{lemma}
\begin{proof}[Proof of Lemma \ref{lemma_vanishing_terms}]
The relevant $T_{\tau}^{ij}$ are 
\begin{align*}
T_{\tau}^{12}=&-\tilde{a}_{T,\tau}^{+}\,\mathcal{D}_{\tau}\,\mathcal{U}_{\epsilon,0}(\tau,-T)\,a_{-T,-T,0}^{-},\\
T_{\tau}^{21}=&-\mathcal{U}_{\epsilon}^*(T,\tau)\,\tilde{a}_{T,T}^{+}\,\mathcal{D}_{\tau}\,a_{\tau,-T,0}^{-},\\
T_{\tau}^{22}=&\mathcal{U}_{\epsilon}^*(T,\tau)\,\tilde{a}_{T,T}^{+}\,\mathcal{D}_{\tau}\,\mathcal{U}_{\epsilon,0}(\tau,-T)\,a_{-T,-T,0}^{-},\\
T_{\tau}^{23}=&-\epsilon\, \mathcal{U}_{\epsilon}^*(T,\tau)\,\tilde{a}_{T,T}^{+}\,\mathcal{D}_{\tau}\,r_0^{-}(\tau,-T),\\
T_{\tau}^{32}=&-\epsilon\,\tilde{r}^{+}(T,\tau)\,\mathcal{D}_{\tau}\,\mathcal{U}_{\epsilon,0}(\tau,-T)\,a_{-T,-T,0}^{-}.
\end{align*}
Note that for $A\in \mathcal{B}(\mathcal{H})$ and $B\in \mathcal{J}_1(\mathcal{H})$,
\begin{equation}\label{eq_norm_inequality}
|\tr(AB)|\leq \Vert A\Vert\Vert B\Vert_1.
\end{equation}
Since $\Vert A\Vert\leq \Vert A\Vert_1$, we will henceforth only consider trace class norms. By computation,
\begin{equation}\label{a_0}
\Vert a_{\tau,-T,0}^{-}\Vert_1=\frac{g}{8\,e_{\tau}^3}
\end{equation}
and
\begin{equation}\label{a_*}
\Vert \tilde{a}_{T,\tau}^{+}\Vert_1=\frac{g}{8\sqrt{1+\gamma_{\tau}^2}e_{\tau}^3}.
\end{equation}
Furthermore
\begin{equation}\label{norm_D}
\Vert\mathcal{D}_{\tau}\Vert_1\leq 4\,e_{\tau},
\end{equation}
see Appendix \ref{app:calc:D}. Hence the statement follows by (\ref{norm_U}), $e_T\rightarrow \infty$ as $T\rightarrow \pm \infty$, and the boundedness of $\gamma_{\tau}$ and of the remainder terms.
\end{proof}
\begin{lemma}\label{lemma_rest_terms_1}
\begin{equation*}
\sup_T\,\epsilon\int_{-T}^T\frac{\gamma_{\tau}}{2}\tr(T_{\tau}^{13}+T_{\tau}^{31})\,\mathrm{d}\tau\leq \, C\,\gamma\,\epsilon^2\,,
\end{equation*}
for some constant $C>0$.
\end{lemma}
\begin{proof}
We have
\begin{equation*}
T_{\tau}^{13}=\epsilon\,  \tilde{a}_{T,\tau}^{+}\,\mathcal{D}_{\tau}\,r_0^{-}(\tau,-T)
\end{equation*}
and 
\begin{equation*}
T_{\tau}^{31}=\epsilon \, \tilde{r}^{+}(T,\tau)\,\mathcal{D}_{\tau}\,a_{\tau,-T,0}^{-}\,.
\end{equation*}
Thus the inequality follows from (\ref{a_0}-\ref{norm_D}) and integration.
\end{proof}
\begin{lemma}\label{lemma_rest_terms_2}
\begin{equation}\label{T_33}
\sup_{T}\,\epsilon\int_{-T}^T\frac{\gamma_{\tau}}{2}\tr(T_{\tau}^{33})\,\mathrm{d}\tau \leq \tilde{C}\,\gamma\, \epsilon^3\,,
\end{equation}
for some constant $\tilde{C}>0$.
\end{lemma}
\begin{proof}
We have
\begin{equation*}
T_{\tau}^{33}=\epsilon^2 \tilde{r}^{+}(T,\tau)\,\mathcal{D}_{\tau}\,r_0^{-}(\tau,-T).
\end{equation*}
While $r^{+}(T,\tau)$ and $r_0^{-}(\tau,-T)$ are uniformly bounded, this alone does not provide the necessary decay of $T_{\tau}^{33}$. Let us thus expand these factors according to (\ref{rest_term_ad_th}), i.e.\ 
\begin{equation}\label{remainder_0}
r_0^{-}(\tau,-T)=b^{-}_{\tau,-T,0}-\mathcal{U}_{\epsilon,0}(\tau,-T)b^{-}_{-T,-T,0}-\int_{-T}^{\tau} \mathcal{U}_{\epsilon,0}(\tau,\zeta)\dot{b}^{-}_{\zeta,-T,0}\,\mathrm{d}\zeta
\end{equation}
and 
\begin{equation}\label{remainder_*}
\tilde{r}^{+}(T,\tau)=\tilde{b}^{+}_{T,\tau}-\mathcal{U}_{\epsilon}^*(\tau,T)\tilde{b}^{+}_{T,T}-\int_{T}^{\tau} \mathcal{U}_{\epsilon}^*(\tau,\zeta)\dot{\tilde{b}}^{+}_{T,\zeta}\,\mathrm{d}\tau,
\end{equation}
where $b^{-}_{\tau,-T,0}$ and $\tilde{b}^{+}_{T,\tau}$ are defined in (\ref{def_b_j}). The first terms on the r.h.s of equations (\ref{remainder_0}) and (\ref{remainder_*}) are bounded by a constant times $e_{\tau}^{-3}$, whereas the second terms are $O(T^{-3})$ (see Appendix \ref{app:calc:b}), whence any term containing them contributes $O(T^{-2})$ to the constant in (\ref{T_33}). 

We may thus pretend that only the last terms are present on the right hand side of (\ref{remainder_0}, \ref{remainder_*}). We observe that $\dot{b}^{-}_{\zeta,-T,0}, \dot{\tilde{b}}^{+}_{T,\zeta}=O(e_{\zeta}^{-3})$ by Appendix \ref{app:calc:b:dot}, so that the two integrals are uniformly bounded in $T$. Moreover, they are $O(\tau^{-2})$ for $\tau \rightarrow -\infty$, respectively for $\tau \rightarrow \infty$, uniformly in $T$. Hence $T_{\tau}^{33}=O(\tau^{-2})$, ($\tau \rightarrow \pm \infty$).
\end{proof}

\end{subsection}


\begin{section}{Extensions}
\label{sec:extensions}
We chose to present our results in the physically most relevant case of an avoided two level crossing. However, a brief inspection of our method shows that it is applicable to more general settings. Let $\mathcal{L}_s = -i[H_s,\cdot] + (\gamma/2) \,\mathcal{D}_s$ be a dephasing Lindbladian with two (among many) stationary projections $P^+_s$ and $P^-_s$. Then the formula (\ref{expansion_p}) for the probability of a transition between these two levels, 
\begin{equation*}
p(\epsilon, \,\gamma, T)=p_c(\epsilon, \,\gamma, T)
 +\frac{1}{2\epsilon}\int_{-T}^T  \gamma \tr((\mathcal{U}_{\epsilon}^*(T,\tau)\,P_T^{+})\, \mathcal{D}_{\tau}\,\mathcal{U}_{\epsilon,0}(\tau,-T)P_{-T}^{-})\,\mathrm{d}\tau,
\end{equation*}
remains unchanged. The coherent ($p_c(\epsilon, \,\gamma, T) := \tr(P_T^{+}\,\mathcal{U}_{\epsilon,0}(T,-T)P_{-T}^{-})$) and incoherent contribution to the tunneling add. Moreover, provided that the error terms in the expansion (\ref{solution_ad_th}) have sufficient decay, the latter is of order $\epsilon$. 

Let us now consider a setting where this is the case. Suppose that $\mathcal{L}_s$ is a minimally degenerate dephasing Lindbladian acting on a finite dimensional space such that $\mathcal{L}_s^{-1} (\mathbb{1}-\mathcal{P}_s)$ is uniformly bounded for $s \in \mathbb{R}$, and that $\mathcal{L}_s$ is three times differentiable. Then Theorem~6 in \cite{AFGG12}, which in particular is a generalization of Theorem~\ref{adiabatic_theorem} w.r.t.\ differentiability conditions, may be applied. We thus arrive at formula (\ref{expansion_p_d}), with $a,\,b$ and $r$ computed with respect to the new Lindbladian. Assume now that $\mathcal{L}_s$ has finite limits $\lim_{s \to \pm\infty} \mathcal{L}_s$ and that on both ends
\begin{equation*}
\frac{\mathrm{d}^j \mathcal{L}_s}{\mathrm{d} s^j} = O(|s|^{-{j-1}}), \quad \mbox{for} \quad j=1,\,2,\,3.
\end{equation*}
The derivatives of $\mathcal{P}_s$, resp. $P^{\pm}_s$, inherit the decay property of the Lindbladian. Thus the remainder terms are uniformly bounded also in $s,s'$; and $\hat{a}_{T,T}$ and $a^{-}_{-T,\,-T,0}$ are of order $T^{-2}$. Hence Lemma \ref{lemma_vanishing_terms} holds true in this situation. Furthermore $T^{13},\,T^{31}$ and $T^{33}$ are integrable and hence contribute by an error $O(\gamma \epsilon^2)$ to the transition probability. It remains to compute the $T^{11}$ contribution. Let $P_s = P_s^+ + P_s^-$, then following the proof of Lemma \ref{lemma_generality_Lindbladian} we see that
\begin{equation*}
\mathcal{L}_s(P_s \rho P_s) = -i[h_s, P_s \rho P_s] - \frac{\gamma_s}{2}[\sqrt{h_s},[\sqrt{h_s}, P_s \rho P_s]]
\end{equation*}
for some Hamiltonian $h_s$ acting non-trivially only on the two levels $P^\pm_s$ and a function $\gamma_s$. The trace of the $T^{11}_s$ term depends only on this reduced Lindbladian and hence is again given by Eq. (\ref{eq_number}) with $e_s$ being the energy gap of $h_s$. In this way we arrive at the formula
\begin{equation*}
p(\epsilon, \,\gamma)=p_c(\epsilon, \,\gamma) + \epsilon \int_{-\infty}^{\infty}\,\frac{\gamma_{\tau}}{1+\gamma_{\tau}^2}\frac{\tr(P_{\tau}^{-}(\dot{P}_{\tau}^{+})^2P_{\tau}^{-})}{e_{\tau}}\,\mathrm{d}\tau +  O(\gamma\epsilon^2).
\end{equation*}
Details shall be presented elsewhere. Sufficient conditions under which the coherent part of the tunneling is exponentially small were given in \cite{Joye}.
\end{section}


\begin{appendix}
\begin{section}{Appendix}
\label{app:calculation:norms}
\begin{subsection}{}
\label{app:calc:b}
We have
\begin{align*}
b^{\pm}_{s,s'}=& \mathcal{L}_s^{-1}\dot{\mathcal{P}}_s\int_{s'}^s \mathcal{T}(s,\tau)\dot{\mathcal{P}}_{\tau}\mathcal{L}_{\tau}^{-1}\dot{P}^{\pm}_{\tau}\,\mathrm{d}\tau+\mathcal{L}_s^{-1}(\mathbb{1}-\mathcal{P}_s)\frac{\mathrm{d}}{\mathrm{d}s}(\mathcal{L}_s^{-1}\dot{P}^{\pm}_s)\\
=&\mp \frac{g^3((i-\gamma_s) E_{s}+h.c.)}{512(1+\gamma_{s}^2)e_s^3}\int_{s'}^s \frac{\gamma_{\tau}}{(1+\gamma_{\tau}^2)e_{\tau}^5}\,\mathrm{d}\tau\\
&\mp \Bigl(\frac{ g(i-\gamma_{s})}{32(1+\gamma_{s}^2)^2e_s^4}(\frac{3s (i-\gamma_{s})}{4\,e_s^2}+\dot{\gamma}_s+ \frac{2\dot{\gamma}_{s}\gamma_{s}(i-\gamma_{s}))}{1+\gamma_{s}^2}) E_{s}+h.c.\Bigr).
\end{align*}
This follows from (\ref{def_b_j}, \ref{P_dot_lin_comb}-\ref{action_dot_P}) and by using
\begin{gather}
\dot{\mathcal{P}}_sP_s^{\pm}=\pm \frac{g}{8e_s^2}(E_s+E_s^*)\,,\nonumber \\
\dot{E}_s=\dot{E}_s^*=-\frac{g}{8e_s^2}(P_s^{+}-P_s^{-})\,.\label{dote}
\end{gather}
Since $\gamma_s$ and $\dot{\gamma}_s$ are bounded, it follows that 
\begin{equation}\label{norm_b}
\Vert b^{\pm}_{s,s'}\Vert_1\leq c \,e_s^{-3},
\end{equation}
and hence $\Vert b^{\pm}_{s,s'}\Vert_1=O(|s|^{-3})$. In particular, (\ref{norm_b}) also holds for $b^{-}_{s,s',0}=b^{-}_{s,s'}(\gamma_s\equiv 0)$ and $\tilde{b}_{s',s}^{+}$. The implication for the latter can be seen by noting that complex conjugation of all functions and changing the sign of the integral in $b_{s,s'}^{+}$ yields $\tilde{b}_{s',s}^{+}$. 
\end{subsection}
\begin{subsection}{}
\label{app:calc:b:dot}
The derivative of the coefficients in $b^{\pm}_{s,s'}$ is bounded by a constant times $e_s^{-3}$ as well, since $\ddot{\gamma}_s$ is bounded. Together with (\ref{dote}) it follows 
\begin{equation}\label{norm_b_dot}
\Vert \dot{b}^{\pm}_{s,s'}\Vert_1\leq c \,e_s^{-3},
\end{equation}
and the corresponding inequality for $\dot{b}^{-}_{s,s',0}$ follows. By Appendix \ref{app:calc:b}, the bound also holds for $\dot{\tilde{b}}_{s',s}^{+}$.
\end{subsection}
\begin{subsection}{}
\label{app:calc:D}
We claim 
\begin{equation}\label{ineq_D}
\Vert \mathcal{D}_{\tau}\Vert_1\leq 4e_{\tau}.
\end{equation}
\begin{proof}
We decompose a general matrix $\rho\in \mathcal{J}_1(\mathcal{H})$ as
\begin{equation*}
\rho=a_{+}P_{\tau}^{+}+a_{-}P_{\tau}^{-}+b_{+}E_{\tau}+b_{-}E_{\tau}^*
\end{equation*}
with $a_{\pm},b_{\pm}\in\mathbb{C}$. We shall show 
\begin{gather}
\Vert b_{+}E_{\tau}+b_{-}E_{\tau}^*\Vert_1=|b_{+}|+|b_{-}|,\label{ineq_D_1}\\
\Vert \rho \Vert_1 \geq |b_{+}|+|b_{-}|,\label{ineq_D_2}
\end{gather}
from which (\ref{ineq_D}) follows in view of 
\begin{equation*}
\mathcal{D}_{\tau}\rho=-4e_{\tau}(b_{+}E_{\tau}+b_{-}E_{\tau}^*)\,,
\end{equation*}
see (\ref{D_on_E}).
If $a_{\pm}=0$ we have $\rho^*\rho=|b_{+}|^2P_{\tau}^{-}+|b_{-}|^2P_{\tau}^{+}$ and $|\rho|=|b_{+}|P_{\tau}^{-}+|b_{-}|P_{\tau}^{+}$, from which (\ref{ineq_D_1}) follows. In the general case, consider the unitary operator
\begin{equation*}
U = \frac{\bar{b}_-}{{|b_-|}} E_\tau + \frac{\bar{b}_+}{{|b_+|}} E^*_\tau\,.
\end{equation*}
We have $\tr(U \rho) = |b_+| + |b_-|$ and hence (\ref{ineq_D_2}) follows from the variational formula 
\begin{equation*}
||\rho||_1 = \sup_{||X||=1} |\tr(X \rho)|\,.
\end{equation*}
\end{proof}
\end{subsection}
\end{section}
\end{appendix}


\textbf{Acknowledgements:} We thank Gian Michele Graf, Yosi Avron and \linebreak Shlomi Hillel for fruitful discussions.
This research was partly supported by the NCCR SwissMAP, funded by the Swiss National Science Foundation.

\end{document}